\documentclass[10pt,conference, twocolumns]{IEEEtran} 

\usepackage{graphicx}
\usepackage{caption}
\usepackage{subcaption}
\usepackage{times}
\usepackage{cite} 
\usepackage[english]{babel}
\usepackage{epsfig} 
\usepackage{amsfonts}
\usepackage{amssymb}
\usepackage{amsmath} 
\usepackage{amsthm}
\usepackage{bm}
\usepackage{amsbsy}
\usepackage{balance} 
\usepackage{latexsym}
\usepackage{txfonts} 
\usepackage{wasysym}
\usepackage{mathbbol}
\usepackage{multirow} 
\usepackage{epstopdf} 
\usepackage[ruled,vlined]{algorithm2e} 
\usepackage{color}  
\theoremstyle{plain}

\newtheorem{lemma}{Lemma}

\newtheorem{property}{Property}

\newtheorem{remark}{Remark}
\usepackage{authblk}

\allowdisplaybreaks
\title{Resource Allocation for NOMA-based LPWA Networks Powered by Energy Harvesting}
\author{Fatma Benkhelifa, Julie McCann\\ 
\small Imperial College London, London, UK\\ 
\small \{f.benkhelifa,j.mccann\}@imperial.ac.uk
}

\begin{document}
\maketitle

\begin{abstract}
{\ In this paper, we explore perpetual, scalable, Low-powered Wide-area networks (LPWA). Specifically we focus on the uplink transmissions of non-orthogonal multiple access (NOMA)-based LPWA networks consisting of multiple self-powered nodes and a NOMA-based single gateway. The self-powered LPWA nodes use the "harvest-then-transmit" protocol where they harvest energy from ambient sources (solar and radio frequency signals), then transmit their signals. The main features of the studied LPWA network are different transmission times-on-air, multiple uplink transmission attempts, and duty cycle restrictions. The aim of this work is to maximize the time-averaged sum of the uplink transmission rates by optimizing the transmission time-on-air allocation, the energy harvesting time allocation and the power allocation; subject to a maximum transmit power and to the availability of the harvested energy. We propose a low complex solution which decouples the optimization problem into three sub-problems: we assign the LPWA node transmission times (using either the fair or unfair approaches), we optimize the energy harvesting (EH) times using a one-dimensional search method, and optimize the transmit powers using a concave-convex (CCCP) procedure. In the simulation results, we focus on Long Range (LoRa) networks as a practical example LPWA network. We validate our proposed solution and we observe a $15\%$ performance improvement when using NOMA. 
\par}
\end{abstract}

\section{Introduction}
{\ Low power wide area (LPWA) networks are being widely used in many applications that span our everyday life such as environmental monitoring, smart homes, agriculture, energy management, health monitoring, etc. LPWA networks (LPWANs) have captured both industrial and research interest due to their long-range communication distances and low power operation, supporting the use of low cost devices. Typically, these systems are characterized by a relatively large number of deployed devices, each with limited resources that need to deliver reliable data rates to, potentially critical, applications. Many standard technologies implement LPWAN such as SigFox, Narrowband Internet of Things (NB-IoT)~\cite{Primer_NBIOT}, Long Range (LoRa)~\cite{loraweb}, Weightless, Ingenu, etc. Most of these technologies rely on the use of different resource blocks (time, frequency, channels, etc.) assigned to different clusters of nodes to allow massive connectivity and reliability among nodes. 
\par}
{\ Serving a huge number of nodes, being bounded by limited resources, poses a challenge to LPWANs while continuing to maintain reliability and sustainability requirements. Non-orthogonal multiple access (NOMA) has been recently proposed as a promising solution to remedy the interference limitations and enable the extension of network capacity and access capability. NOMA allows nodes to share nonorthogonal resources to improve spectral efficiency while allowing some degree of multiple access interference at receivers \cite{tutnoma}. NOMA multiplexing falls under two types: power domain multiplexing and code domain multiplexing. Power domain multiplexing based NOMA assigns different powers to different nodes according to their channel conditions, where successive interference cancellation (SIC) is performed at the receiver side. In uplink NOMA, SIC decodes the received signals successively by first ordering the nodes according to their channel gains, starts by decoding the most powerful signal, subtracts it from the remaining signals, and repeats the same process on the resulting signals, until it decodes the less powerful signal after subtracting all the interfering signals.
\par}
{\ Despite its potential to combat interference, only few research works have investigated the incorporation of NOMA in LPWANs \cite{nbiotnoma1,nbiotnoma2,kaihan}. In \cite{nbiotnoma1}, connection density was maximized while jointly optimizing subcarrier and power allocation for NOMA-based NB-IoT for machine-to-machine communications (M2M). Optimal and near optimal solutions were proposed using mixed integer linear programming and difference of convex programming approaches. Also, suboptimal solutions were proposed using low-complexity heuristic algorithms. In \cite{nbiotnoma2}, the total throughput was maximized for a NOMA-based NB-IoT for M2M while optimizing the resource allocation and NOMA node clustering satisfying transmit power constraints and quality of service requirements. In \cite{kaihan}, the minimum throughput rate was maximized for NOMA-enabled LPWAN while optimizing the channel allocation, the transmission time allocation and the power allocation. Here, only one transmission attempt was considered per node independently of the spreading factor (SF) value. 
\par}
{\ Energy consumption is another major challenge in LPWANs where most devices are generally equipped with batteries with limited capacity - it may be difficult or expensive to replace batteries, especially in difficult to access or dangerous locations. In order to improve the energy efficiency of LPWAN devices, energy harvesting (EH) is a promising solution that provides continuous energy to ensure operation sustainability for LPWANs \cite{EHWSN}. Indeed, energy can be harvested from many sources: wind, solar, radio frequency (RF) signals \cite{fundamsurveryEH}, etc. Few papers have considered energy harvesting capabilities for LPWAN networks. For example, in \cite{EHlpwan}, a control framework for EH nodes connected through multiple access channels was developed. A throughput-optimal policy was derived for the genie-aided case of non-causal knowledge of the number of active nodes and a Bayesian estimation approach was put forward for the practical cases. In \cite{EHlpwan2}, centralized and distributed resource and task scheduling was studied for EH-based IoT network while minimizing the energy consumption satisfying a minimum average data rate and minimum task performing rate. The IoT devices are powered by a hybrid access point (HAP), harvesting from both on-grid and renewable energy sources. In~\cite{ehlora}, a battery-less LoRa wireless sensor was proposed that monitors road conditions powered by an electromagnetic energy harvester based on a Halbach configuration. In \cite{myICCpaper}, an EH-based LoRa network was studied where devices were harvesting energy from ambient sources and using the energy to transmit once. The minimum uplink throughput rate was maximized while optimizing the EH time allocation, the SF allocation and power allocation. \cite{myICCpaper} highlighted that interference was the main limitation of the system performance and not mainly the energy scarcity.
\par}
{\ However, to the best of our knowledge, incorporating both the interference cancellation techniques and energy harvesting capabilities for LPWA networks weren't studied before in the literature. In addition, none have examined the time averaged sum of transmission throughput rate for a large time window where multiple transmission attempts are possible. Given this context, our work is the first to investigate the uplink NOMA-based LPWA network with nodes harvesting energy from ambient RF or solar sources. In our scheme, the gateway employs NOMA to decode the overlapping signals received from the simultaneous transmissions of LPWA nodes. Our objective is to maximize the temporal averaged sum of uplink transmission rates while optimizing the EH time duration, the transmission time duration, and the transmit powers subject to maximum transmit power at each node and subject to the availability of harvested power in the storage unit. 
\par}

\section{System Model}
{\ We consider a NOMA-based LPWA network where multiple $U$ LPWA nodes are communicating with a single gateway during a time window $T$. We consider a star-of-stars topology where the $U$ LPWA nodes are uniformly distributed in a circle of radius $R$ centred around the gateway, and the gateway relays the data  transmitted from the LPWA nodes to the server. For the EH protocol, the LPWA nodes use the harvest-then-transmit (HTT) protocol where they harvest energy from ambient external sources (RF or solar), during the EH time duration, then transmit their information data using the harvested power during their transmission time-on-airs, and finally store the remaining energy after transmission in rechargeable batteries. The gateway is equipped with SIC to decode the nodes sharing the same resources.
\par}
{\ We subdivide the time window $T$ into $K$ equal time slots of duration $T_{slot}$. At the time slot $k=1,\dots,K$, the channel between the $n$'th LPWA node and the gateway is modelled as a Rayleigh fading channel with path loss as
		\begin{align}
			g_{n,k} &= h_{n,k} \left(d_{n,k}\right)^{-\alpha},
		\end{align}
    where $n=1,\dots,U$, $h_{n,k}$ is the small scale fading that is exponentially distributed with unit mean, $d_{n,k}$ is the distance between the n'th LPWA node and the gateway, and $\alpha$ is the path loss exponent. Without loss of generality, the channel gains are sorted in a descending order, i.e. $g_{1,k}\geq g_{2,k}\geq \dots \geq g_{U,k}$. 
    \par}
\subsection{LPWAN Physical Layer}
{\ For the LPWAN physical layer, the nodes have access to a number $M$ of possible transmission time-on-air (ToA). The $n$'th LPWA node is assigned to one out of M possible ToAs given by
    \begin{align}
		T_{a,n} &= nb_n \hspace{1mm} T_{s,n} \in \mathcal{T}_a =\left\{ t_{a,i}= nb_i t_{s,i}; i=1,\dots,M \right\},
	\end{align}
	where $nb_n$ is the number of symbols, $T_{s,n}$ is the symbol duration, $\mathcal{T}_a$ is the set of possible ToAs of discrete cardinality equal to $M$, i.e. $\vert\mathcal{T}_a\vert=M$. 
\par} 
{\ In addition, the LPWA network is subject to duty cycle restrictions. For example, European frequency regulations impose duty cycle restrictions on some LPWA bands, such as LoRa operating on $868$ MHz sub-bands to have duty cycle restrictions of either $1\%$ or $10\%$~\cite{doppler}. Let us denote by $d$ the duty cycle which can be equal either to $1\%$ or $10\%$. This means that each LPWA node should stay silent $(1-d)\%$ of the packet duration once the transmission happens over one channel. Subsequently, its time-off per channel and the total packet duration are expressed as $T_{off,n}= \frac{1-d}{d}  T_{a,n}$ and $T_n = \frac{T_{a,n}}{d}$, respectively. Let $\text{BW}$ denote the bandwidth of the network. 
\par}
{\ In order to simplify our analysis, we consider the following two properties verified by our LPWA network:
    \begin{property}\label{propto1}
        We assume that the transmission ToAs verify $t_{a,1}< t_{a,2} < \dots < t_{a,M}$ and $\sum_{i=1}^M t_{a,i} < t_1$. The first inequality means that the LPWA nodes with higher $T_{a,n}=t_{a,i}$, for $i>1$, cannot transmit again until the nodes with the least $t_{a,1}$ transmit. The second inequality means that, during the time period $t_1$, all LPWA nodes are transmitting once and no more than that.
    \end{property}
    \begin{property}\label{propto2}
        We assume that all LPWA nodes have the same number of symbols. Subsequently, the ToAs verify $t_{a,i}=2 t_{a,i-1}$, $\forall i=2,\dots,M$. Same property is verified by $t_{off,i}=2t_{off,i-1}$ and $t_i=2 t_{i-1}$, $\forall i=2,\dots,M$.
    \end{property} 
    Satisfying properties \ref{propto1} and \ref{propto2}, we choose that the slot duration is exactly equal to the lowest packet duration $T_{slot}=t_1$, without loss of generality. This way, the LPWA nodes send only at time slots multiple of $\frac{T_{a_n}}{t_{a,1}}$. Let $\rho_n(k)$ be a binary variable denoting if the $n$'th LPWA node is transmitting or not at the $k$'th time slot as
    \begin{align}
        \rho_n(k) &=\begin{cases}
        1, &\mbox{ if HTT mode i.e. }  \mod(k,\frac{T_{a,n}}{t_{a,1}})=0,\\
        0, &\mbox{ if EH mode i.e. }  \mod(k,\frac{T_{a,n}}{t_{a,1}})\neq 0,
        \end{cases}
    \end{align}
    and $\mu_n(k)$ denotes the number of uplink transmission attempts performed until the time slot $k$, i.e. $\mu_n(k)=\sum_{j=1}^k \rho_n(k)$. Thus, the $n$'th LPWA node sends its $m_n$ transmission attempts at time slot $k_n = m_n \frac{T_{a,n}}{t_{a,1}}=\rho_n(k) k$, with $m_n=1,\dots,a_n$ where $a_n$ is the maximum number of uplink transmission attempts during the time window $T$ verifying $a_n = d \left\lfloor \frac{T}{T_{a,n}} \right\rfloor$. 

\subsection{Harvest-Then-Transmit Protocol}
{\ Powered via ambient sources such as RF or solar, the nodes use the HTT protocol, harvesting the required energy first before transmitting their data. 
For a given time slot $k=1,\dots,K$, each LPWA node $n$ can be either in EH mode (i.e. $\rho_n(k)=0$) or in HTT mode (i.e. $\rho_n(k)=1$ or equivalently $k\propto \frac{T_{a,n}}{T_{slot}}$), depending on its transmission ToA assigned by the gateway. 
For the EH mode, the LPWA node harvests energy up to the time slot duration $T_{slot}$. For the HTT mode, the LPWA node harvests energy up to $\left(T_{slot}-T_{a,n}\right)$ and then transmits its data to the gateway during its ToA $T_{a,n}$. Let $\tau_{e,n}(k)$ be the EH duration at the $n$'th LPWA node during the $k$'th time slot. To comply with the HTT protocol and duty cycle restrictions, $\tau_{e,n}(k)$ is subject to the constraint $0 \leq \tau_{e,n}(k) \leq \tau_{e,n,max}(k)$, where the maximum EH time for the $n$th LPWA node during the $k$'th time slot is $\tau_{e,n,max}(k)=T_{slot}-\rho_n(k) T_{a,n}$. As said previously, the LPWA nodes are harvesting energy from RF or solar sources. 
\par} 
\begin{paragraph}{RF Source}
{\ We harvest from ambient $N_b$ power beacons (PBs) located at distance $d_{b,n}$ from the $n$'th LPWA node and transmitting with power $P_b$. The RF harvested energy per time unit for the $n$'th LPWA node during the $k$'th time slot is modelled as
    \begin{align}
        E_n(k) &= \sum_{b=1}^{N_b} \Psi\left(P_{rec}(n,b)(k) \right),
    \end{align}
    where $P_{rec}(n,b)(k)= P_b h_{b,n,k}  d_{b,n}^{-\alpha_b}$ is the received power from the $b$'th PB at the $n$'th LPWA node, $h_{b,n,k}$ is the channel between the $b$'th PB and the $n$'th LPWA node at $k$'th slot, $\alpha_b$ is the path loss exponent between PBs and LPWA nodes, and the function $\Psi(\cdot)$ is defining the relationship between the harvested power and the received power. If the nonlinear (NL) EH model in \cite{elena} is considered, $\Psi(\cdot)$ is defined as $\Psi(x)  = \frac{\beta(x) -M \Omega}{1-\Omega}$, where $\beta(x)= \frac{M}{1+e^{-a(x-b)}}$, $\Omega=\frac{1}{1+e^{ab}}$, $M$ is the maximum harvested energy, and $a$ and $b$ are experimental parameters which reflect the nonlinear charging rate with respect to the input power and the minimum required turn-on voltage for the start of current flow through the diode, respectively \cite{elena}. If the linear model is considered, $\Psi(\cdot)$ is simply given by $\Psi(x) = \zeta_{RF} x$, where $\zeta_{RF}\in\left[0,1\right]$ is the RF conversion efficiency.
\par}
\end{paragraph}
\begin{paragraph}{Solar Source}
{\ The solar harvested energy per time unit for the $n$'th LPWA node during the $k$'th time slot is given by \cite{greg}
    \begin{align}
        E_n(k) &= \zeta_S A_{area} G_{bn} cos(\theta),
    \end{align}
    where $\zeta_S$ is the solar conversion efficiency, $A_{area}$ is the solar panel area, $G_{bn}$ is the solar irradiation hitting a titled solar plane, and $\theta$ is the angle of incidence of the sun rays on the titled plane.
\par} 
\end{paragraph}
{\ Subsequently, the harvested energy and the corresponding harvested power for the $n$ LPWA node during the $k$'th time slot are given by
    \begin{align}
        E_{h,n}(k) &= \tau_{e,n} E_n(k), &\text{ and }&&
        P_{h,n}(k) &= \frac{\tau_{e,n} E_n(k)}{T_{a,n}},
    \end{align}
    respectively. Furthermore, the LPWAN nodes are equipped with rechargeable battery that stores the excess of energy that wasn't used to transmit. Hence, the available power to use for the $n$'th LPWA node at the $k$ time slot is the cumulative sum of harvested power from the previous time slots and the current time slot minus the sum of transmit powers during the previous time slots.
    \begin{align}
    	P_{a,n}(k) &= \sum_{j=1}^{k} P_{h,n}(j)-\sum_{j=1}^{k-1} \rho_n(j) p_n(j),
    \end{align}
    where $p_n(j)$ is the transmit power of the $n$'th LPWA node at the $j$'th time slot subject to the maximum transmit power $P_t$ and to the causality conditions as follows
    \begin{align}
        &0 \leq p_n(k) \leq P_{t}, &\text{ and }&
    	&0 \leq p_{n}(k) \leq P_{a,n}(k).
    \end{align}
    These mean that the $n$'th LPWA node cannot transmit more than the maximum power $P_t$ and more than the available power $P_{a,n}(k)$ at the $k$'th time slot.
\par}
\section{Sum-Rate Maximization Problem}
{\ Having access to different transmission ToAs and harvesting energy until one transmission attempt, packet collision overlap time can happen between LPWA nodes and that is depending on their TAs and their EH times. Considering the general expression of packet collision time between LPWA nodes in \cite{myICCpaper}, the packet collision time for this LPWA network is formulated as
    \begin{align}
        col_{n,m}(k) &= \begin{cases}
					0, &\hspace{-44mm}\mbox{ if }  \vert \tau_{e,n}(k) - \tau_{e,m}(k) \vert\geq   \min\left(T_{a,n},T_{a,m}\right),\\
					\min\left(T_{a,n},T_{a,m}\right)-\vert \tau_{e,n}(k) - \tau_{e,m}(k) \vert,\\ &\hspace{-44mm}\mbox{ if }  \vert \tau_{e,n}(k)  - \tau_{e,m}(k) \vert <  					\min\left(T_{a,n},T_{a,m}\right).
					\end{cases}
    \end{align}
    which implicitly assumes that $\left(\tau_{e,n}(k) - \tau_{e,m}(k)\right)$ has the same sign as $\left(T_{a,n}-T_{a,m}\right)$. 
\par}
\subsection{NOMA-based SINR}
{\ As seen so far, there are so many sources of interference that can occur at the gateway, it makes it hard for the gateway to decode the received data properly. The NOMA technique uses an additional dimension, i.e. the power domain, and superposes the LPWA nodes sharing the same resources (i.e. time, frequency, spreading code, etc.). The decoding of the nodes sharing the same resources happens with the virtue of the suboptimal multiuser detector which is SIC. This technique means that the LPWA nodes are ordered in a descending way according to their channel gains. The most powerful LPWA node is decoded first while treating all the remaining nodes as noise. Then its decoded signal is substituted from all the remaining nodes. The second powerful node is decoded. Then, its decoded signal is substituted from the remaining nodes. This process continues successively until the last node which is decoded with zero interference. Subsequently, the NOMA-based signal-to-interference plus noise ratio (SINR) of the $n$'th LPWA node at the $k$'th time slot is written as 
        \begin{align}
        	\gamma_n^{col,NOMA}(k) &= \frac{p_n(k) g_n(k) }{ \sum_{m = n +1}^U \rho_m(k) \eta_{n,m} (k) p_m(k) g_m(k) + \sigma^2},
	    \end{align}
	    where $\sigma^2$ is the variance of the additive white Gaussian noise (AWGN) at the LPWAN gateway, $\eta_{n,m}(k)=\frac{col_{n,m}(k)}{T_{a,n}} \xi_{n,m}$, and $\xi_{n,m}$ is the correlation factor between the coded waveforms for LPWA nodes $m$ and $n$. Note that $\xi_{n,m}=1$ if $m=n$ and $0\leq\xi_{n,m}<1$ if $m \neq n$. At time slot $k=1,\dots,K$, the instantaneous transmission rate of $n$'th LPWA node is given by 
    	\begin{align}
	    	R_n^{NOMA}(k) &=   \rho_n(k) \log_2\left( 1 +	\gamma_n^{col,NOMA}(k)\right), 
    	\end{align}
	and the temporal average of throughput rate of LPWA node $n$ for the time period $T$ can be written as
    	\begin{align}
        	R_n^{NOMA} &=  \frac{T_{a,n}}{T}\sum_{k =1}^K\rho_n(k) \log_2\left( 1 + \gamma_n^{col,NOMA}(k)\right). 
    	\end{align}
\par}
\subsection{Problem Formulation}
{\ In this paper, we propose to maximize the time-averaged sum rate of all LPWA nodes subject to energy harvesting availability and the use of NOMA technique at the gateway which can be formulated as
    \begin{subequations}
		\begin{align}
		\underset{T_{a,n}, \tau_{e,n}(k), p_n(k)}{\max }\hspace{1mm}& \sum_{n=1}^U R_n^{NOMA},\label{objfun}\\
		\text{s.t. }
		& \text{ C$_{1}$: } 0 \leq p_{n}(k) \leq P_t,  \\
		& \text{ C$_{2}$: } 0 \leq  \sum_{j=1}^{k } \rho_n(j)p_n(j) \leq \sum_{j=1}^k \frac{\tau_{e,n}(j) E_n(j)}{T_{a,n}}, \\ 
		& \text{ C$_{3}$: } 0  \leq \tau_{e,n}(k) \leq \tau_{e,n,max}(k),\\
		& \text{ C$_4$: } 1 \leq \sum_{j=1}^{k } \rho_n(j) \leq a_{n},  T_{a,n} \in\mathcal{T}_a,
		\end{align}\label{mainProbGenT}
	\end{subequations}	
	where the constraint C$_1$ is due to the maximum transmit power constraint at each LPWA node, the constraint C$_2$ is because each LPWA node cannot transmit with a power greater than the available power, the constraint C$_3$ is to respect the maximum EH time, and the constraint C$_4$ is due to the different ToA assignment at each LPWA node. 
	This sum-rate optimization problem (\ref{mainProbGenT}) is a nonconvex problem since it is a mixed-integer programming problem and the objective function is nonconcave due to the interference from colliding LPWA nodes. Hence, the optimal solution is computationally expensive. For example, the exhaustive search of only the ToA assignment has at least a complexity of the order of $U^{KM} N_{\epsilon}^2$, where $N_{\epsilon}$ is the complexity of the one-dimensional search method. In order to deal with its high complexity, we propose a low complex solution that can be implemented in IoT devices. We decouple the main problem into three sub-problems where each optimization variable is optimized separately. First, we assign the ToAs while assuring either fairness or unfairness between LPWA nodes. Second, we optimize the EH time using a one-dimensional exhaustive search. Finally, we optimize the transmit powers for the given SF and EH time.
	\par}
\subsection{ToA Allocation Scheme}
{\ First, we propose to optimize the ToAs assignment. 
Since the exhaustive search is of complexity $U^{K M} N_{\epsilon}^2$, we propose to either fairly or unfairly assign ToAs between LPWA nodes. We assume that the gateway has a minimum receiver sensitivity to satisfy, namely $\xi_{min}$. We assume that all nodes are transmitting with their maximum powers (satisfying both constraints C$_1$ and C$_2$), namely $P_{n,max}$. To further simplify, we perform the ToA assignment only once based on the resources available during the first time slot and we keep the ToA assignment constant during the $K$ time slots. The ToA assignment happens as follows:
    \begin{itemize}
        \item Compute the RSSIs of all LPWA nodes at the beginning of the first time slot.
        \item Find the number $U_a$ of active LPWA nodes having RSSI $>\xi_{min}$.
        \item Divide the $U_a$ ordered LPWA nodes into $M$ groups of size $k_i$, $i=1,\dots,M$, according to their RSSIs. \begin{itemize}
            \item Unfair SF allocation: equally divides the LPWA nodes into $M$ groups, with $k_i=\frac{U_a}{M}$.
            \item Fair SF allocation: fairly divides the LPWA nodes into $M$ groups, such that $k_i t_i=k_j t_j$. Hence, the size of each group is $k_i= \frac{U_a}{t_i \sum_{j=1}^M \frac{1}{t_j}}$.
        \end{itemize} 
    \end{itemize}
\par}
\subsection{EH Allocation Scheme}
{\ Next, for a given ToA assignment, we propose to optimize the EH time allocation. First, recall that we are using the HTT protocol where there is no data transmission during the time when the LPWA node is harvesting energy. Also, the EH time is only present in the two constraints C$_{2}$ and C$_{3}$ in (\ref{mainProbGenT}). Since the constraint C$_2$ depends on the previous time slots, we propose to proceed successively where we start to optimize $\tau_{e,n}(1)$, then optimize successively $\tau_{e,n}(k)$ given $\tau_{e,n}(j)$ for $j=1,\dots,k-1$. For a given time slot, each LPWA node can be in EH mode or HTT mode. Let us denote by $\tilde{\tau}_{e,n,1}(k)=\frac{P_t T_{a,n}}{E_n(k)} \mu_n(k)- \sum_{j=1}^{k-1} \frac{\tau_{e,n}(j) E_n(j)}{E_n(k)}$, and $\tilde{\tau}_{e,n,2}(k)=\min\left(\tilde{\tau}_{e,n,1}(k), \tau_{e,n,max}(k)\right)$. 
	\begin{lemma}
	In EH mode, the optimal EH time is exactly equal to the slot duration. 
	In HTT mode, the optimal EH time depends on the monotonicity of the collision time. If the collision time is a monotonic increasing function with respect to the EH time, the optimal EH time is $\tilde{\tau}_{e,n,2}(k)$. If the collision time is a monotonic decreasing function with respect to the EH time, the optimal EH time is $\tau_{e,n,max}(k)$. Otherwise, the optimal EH time is obtained by a line search method in $\left(0, \tau_{e,n,max}(k)\right]$.
	\end{lemma}
	\begin{proof}
	For the EH mode (i.e. $\rho_n(k)=0$), there is no data transmission during that time slot and the constraint C$_2$ drops. So, the EH time is exactly equal to the time slot duration. 
	For the HTT mode (i.e. $\rho_n(k)=1$), we have two possible cases because of the maximum transmit power constraint in C$_1$: either we are harvesting more than what we need (i.e. $\sum_{j=1}^k \rho_n(j) P_t \leq \sum_{j=1}^k  \frac{\tau_{e,n}(j) E_n(j)}{T_{a,n}}$), or we are harvesting less than what we need (i.e. $\sum_{j=1}^k \rho_n(j) P_t \geq \sum_{j=1}^k \frac{\tau_{e,n}(j) E_n(j)}{T_{a,n}}$). 
	If $\sum_{j=1}^k \rho_n(j) P_t \leq \sum_{j=1}^k \frac{\tau_{e,n}(j) E_n(j)}{T_{a,n}}$, the constraint C$_2$ is trivially satisfied and the EH time should satisfy
	\begin{align}
	    \tilde{\tau}_{e,n,1}(k)  \leq  \tau_{e,n}(k)  \leq \tau_{e,n,max}(k).
	\end{align}
	If the collision time is a monotonic decreasing function with respect to the EH time, the optimal EH time is $\tau_{e,n,max}(k)$. 
	If the collision time is a monotonic increasing function with respect to the EH time (or independent of the EH time), the optimal EH time is $\tilde{\tau}_{e,n,1} (k)$. Otherwise, the optimal value is between $\tilde{\tau}_{e,n,1} (k)$ and $\tau_{e,n,max}(k)$. If $\sum_{j=1}^k \rho_n(j) P_t \geq \sum_{j=1}^k \tau_{e,n}(j) \frac{E_n(j)}{T_{a,n}}$, this means that we are not harvesting enough what we need ($P_t$). The EH time at the $k$'th time slot should satisfy
	\begin{align}
	   0 \leq \tau_{e,n}(k) \leq \min\left(\tilde{\tau}_{e,n,1}(k),\tau_{e,n,max}(k)\right)\leq \tau_{e,n,max}(k).
	\end{align}
	If the collision time is a monotonic increasing function with respect to the EH time (or independent of the EH time), the optimal EH time is $\min\left(\tilde{\tau}_{e,n,1} (k),\tau_{e,n,max}(k)\right)$. Otherwise, the EH time is obtained by a one-dimensional search method in $\left(0,\min\left(\tilde{\tau}_{e,n,1} (k),\tau_{e,n,max}(k)\right)\right]$. 
	\end{proof}
\par}
\subsection{Power Allocation Scheme}
{\ Given the ToA and EH allocation, we propose to investigate the power allocation for the $U_a$ LPWA nodes during the $K$ time slots. The objective function is nonconcave due to the cumulative interference in the denominator which makes it hard to solve the sum rate optimization problem. Thus, we propose a lower bound on the instantaneous transmission based on the concave-convex procedure (CCCP) for which we can obtain an optimal solution. The CCCP uses a first-order Taylor approximation to derive a concave lower bound on the objective function which transforms the original optimization problem into an approximated convex problem. The procedure is as follows. First, we derive the lower bound on the instantaneous transmission rate of LPWA node $n$ during time slot $k$ at $p_n(k)=\hat{p}_n(k)$: 
	\begin{align}
	R_n^{NOMA}(k) &\geq R_{n,1}^{NOMA}(k) + \hat{R}_{n,2}^{NOMA}(k),
	\end{align} 
	where $R_{n,1}^{NOMA}(k)=\log_2\left(  \sum_{m=n}^U \frac{\rho_m(k) \mu_{n,m}(k)}{\sigma^2} p_m(k )  g_m(k)  + 1  \right) - \frac{1+ \sum_{m=n+1}^U  \frac{\rho_m(k) \mu_{n,m}(k)}{\sigma^2} p_m(k )  g_m(k ) }{1 + \sum_{m=n+1}^U \frac{\rho_m(k) \mu_{n,m}(k)}{\sigma^2} \hat{p}_m(k)  g_m(k) }$ is concave, and $\hat{R}_{n,2}^{NOMA}(k)=1-\log_2\left(\sum_{m=n+1}^U  \frac{\rho_m(k) \mu_{n,m}(k)}{\sigma^2} \hat{p}_m(k )  g_m(k)  + 1  \right)$ is independent of $p_n(k)$. 
	Subsequently, for a given $\hat{p}_m(k)$, problem (\ref{mainProbGenT}) simplifies to the equivalent convex optimization problem:
	\begin{subequations}
		\begin{align}
		\underset{ p_n(k)}{\max }\hspace{1mm}& \sum_{n=1}^U R_{n,1}^{NOMA}(k) + \hat{R}_{n,2}^{NOMA}(k),\\
		\text{s.t. }
		& \text{ C$_{1}$-C$_{2}$}. 
		\end{align}\label{mainProbGenTeq}
	\end{subequations}
	The value of $\hat{p}_m(k)$ used to solve (\ref{mainProbGenTeq}) is the value of $p_n(k)$ from the previous iteration and is initialized at iteration=0 by a feasible solution to (\ref{mainProbGenT}).
\remark{ If $\frac{T_{a,n}\rho_n(k)}{T}$ is constant, (\ref{objfun}) simplifies to 
       \small{ \begin{align}
        	\sum_{n=1}^U R_n^{NOMA} &=  \sum_{k =1}^K \frac{T_{a,n}\rho_n(k)}{T} 
        	\log_2\left( \sum_{m = 1 }^U \frac{\rho_m(k) \mu_{n,m}}{\sigma^2} (k) p_m(k) g_m(k) + 1 \right) . 
    	\end{align}}
    whose optimal solution is exactly given by the maximum power $P_{n,max}$ satisfying the constraints C$_1$ and C$_2$. 
    The condition $\frac{T_{a,n}\rho_n(k)}{T}$ being the same for all $n$'s is not always verified by our LPWA network, but it happens for example in perfect orthogonality case where only LPWA nodes having the same ToA interfere with each other. 
}\par}

\section{Numerical Results}
{\ In this section, we present some simulation results to validate our proposed solution. We choose LoRa network as an example of LPWA network. For our purposes, LoRa operates in the unlicensed $868$ MHz frequency and uses the chirp spread spectrum (CSS) modulation for which there are $6$ spreading factors (SFs) corresponding to $6$ ToAs ($M=6$) and the Duty cycle imposed is either $1\%$ or $10\%$. In the following simulations, we consider $d=1\%$, $BW=125$ kHz, $K=4$, and SFs having integer values from $7$ to $12$. The noise variance is defined as $\sigma^2= -174+\text{NF} + 10\log_{10}(\text{BW})$ in dBm, where $\text{NF}$ is the noise figure equal to $6$ dB. The path-loss exponent for both the power transfer and the information transfer links is $3$. The maximum transmit power for all LoRa nodes is chosen equal to $P_t=14$ dBm. In all figures, we are plotting time-averaged sum of uplink transmission throughput for two types of EH sources: RF source and solar source. For the RF EH source, we consider the nonlinear EH model with parameters $N_b=3$, $P_b=0,1$W, $a=1500$, $b=0.0022$, and $M=24$ mW which were shown in \cite{elena} to fit the experimental data in \cite{datael}. For the solar EH source, we consider $\zeta_S=15\%$, $A_{Area}=58mm\times 58mm$, $G_{bn}=1000$W/m$^2$, $\theta\in[0,\frac{\pi}{2}]$ \cite{greg}. In all figures, three different interference scenarios are plotted: "no interference", "co-SF interference", and "co-SF $\&$ inter-SF interference".
\par}
\begin{figure*}[t]
    \centering
    \begin{subfigure}[b]{0.32\textwidth}
        \includegraphics[scale=0.17]{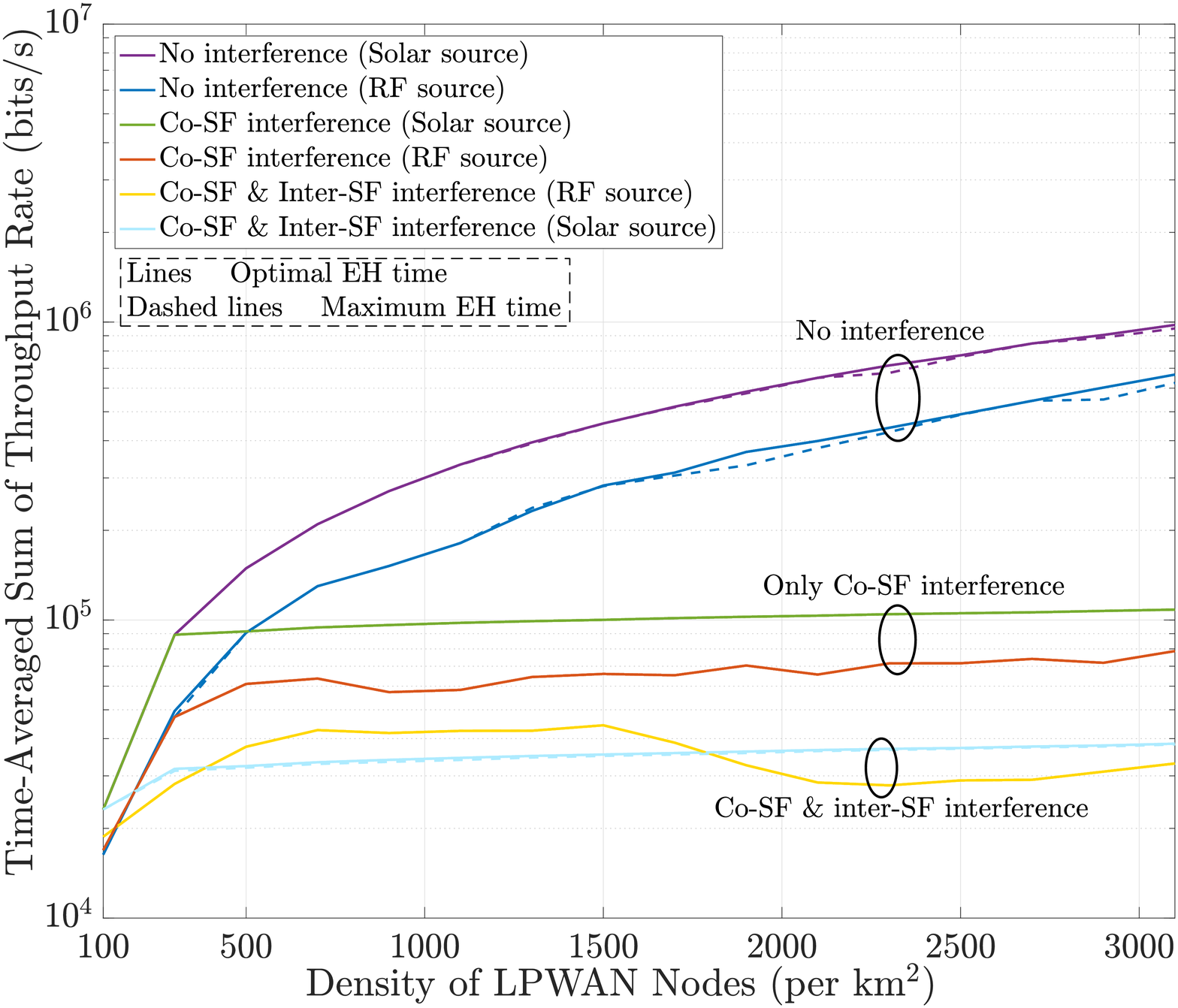}
        \caption{Optimal and maximum EH time}
        \label{subFigEHopt}
    \end{subfigure} \hspace{1mm}
    \begin{subfigure}[b]{0.32\textwidth}
        \includegraphics[scale=0.17]{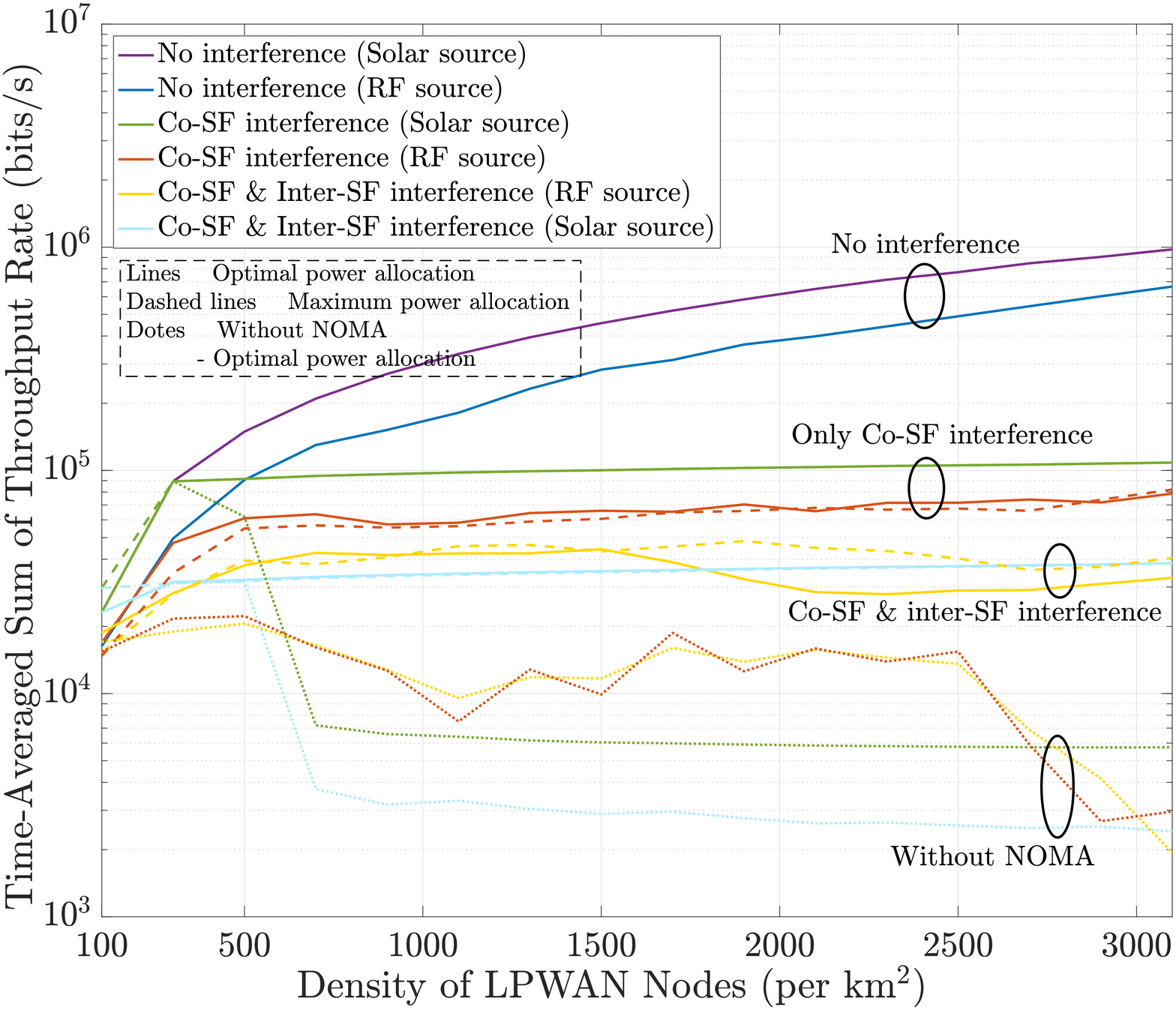}
        \caption{Optimal and maximum power}
        \label{subFigPopt}
    \end{subfigure} \hspace{1mm}
    \begin{subfigure}[b]{0.32\textwidth}
        \includegraphics[scale=0.17]{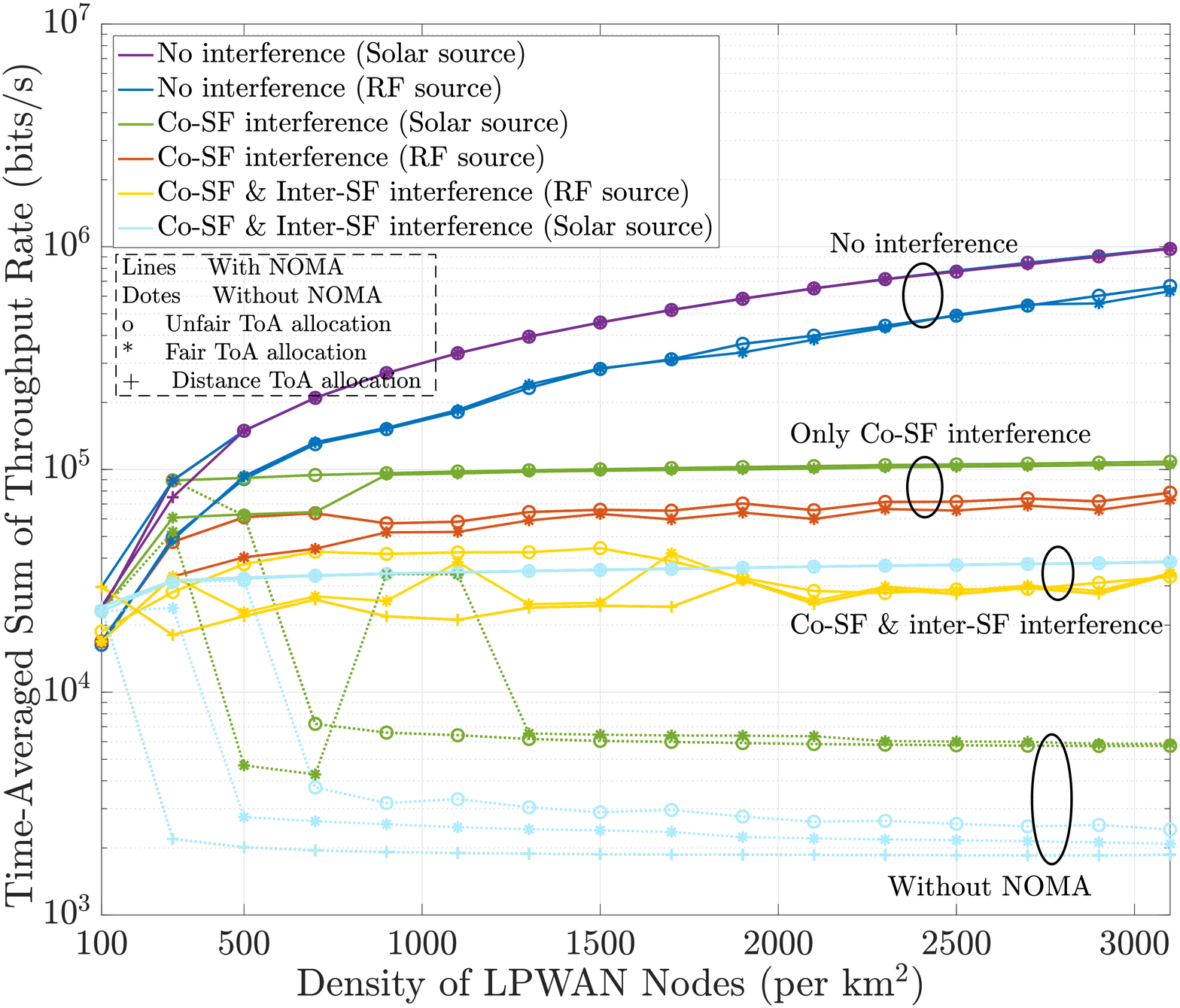}
        \caption{Different ToA allocations}
        \label{subFigSFalloc}
    \end{subfigure} 
    \caption{Time-averaged sum of uplink transmission throughput versus the density of LPWA nodes per km$^2$}
    \label{Figtot}
\end{figure*} 
{\ In Fig. \ref{subFigEHopt}, we have plotted the time-average sum of uplink transmission throughput rate versus the density of LPWA nodes per km$^2$ with unfair ToA allocation and optimal power allocation. We have compared the optimal EH time (lines) versus the maximum EH time $\tau_{e,n,max}(k)$ (dashed lines). For the "no interference" case, the time-average sum of throughput rate with optimal EH time is higher than the one with the maximum EH time for density greater than $1500$ nodes/km$^2$. Moreover, we have a closer match for higher densities for both the "only co-SF interference" and the "co-SF $\&$ inter-SF interference" cases. Another observation is that the time-average sum of throughput rate with solar source is always outperforming the one with RF source. 
\par} 
{\ In Fig. \ref{subFigPopt}, we have plotted the time-average sum of uplink transmission throughput rate versus the density of LPWA nodes per km$^2$ with unfair ToA allocation and optimal EH time. We compare the optimal power allocation versus the maximum power allocation. We also compare this to the no NOMA case representing the commonly evaluated LPWAN performance. First, the use of NOMA in LPWANs demonstrates a clear advantage. For the solar source, we can see an improvement of more than $15$ times in the time-average sum of uplink transmission throughput rates when using NOMA for both "only co-SF interference" and "co-SF $\&$ inter-SF interference" cases. In addition, we can see how much the gap between the "no interference" case and the "only co-SF interference" case, reduces compared to the gap without NOMA. Hence, incorporating NOMA at the gateway helps reduce the impact of the co-SF interference on the time-averaged sum rate. Moreover, both the maximum power and optimal power using CCCP give good performance.
\par} 
 
{\ In Fig. \ref{subFigSFalloc}, we have plotted the time-average sum of uplink transmission throughput rate with optimal EH time and optimal power allocation. We have compared the different ToA allocation schemes: unfair, fair, and distance-based. The distance-based ToA allocation assigns the ToAs depending on the distance between the LPWA nodes and the gateway~\cite{canlorascale}. First, the distance-based ToA allocation has the poorest performance. Moreover, fairly or unfairly allocating LPWA nodes end up to have close time-average sum of the throughput rate.
\par} 

\section{Conclusion}
{\ In this paper, we optimized the time-average sum of uplink transmission throughput rates for NOMA-based LPWA networks powered by ambient energy for a given time window $T$. We optimized the ToA allocation for both fair and unfair scenarios, we optimized the EH allocation using a one-dimensional search method, and then we optimized the power allocation using an approximated convex problem using CCCP procedure. We have seen that the maximum transmit power allocation is the optimal solution where perfect orthogonality applies.
Through simulation results, we have observed that the time-average sum of uplink rate improves significantly using NOMA and either of fair or unfair allocation exhibits similar performance in terms of the time-average sum of the throughput rate.
\par}

\bibliographystyle{IEEEtran}
\bibliography{clean_references}

\end{document}